%% file: main.tex
\documentclass[sigconf,natbib=true,anonymous=false,screen]{acmart}

\acmSubmissionID{766}

\input{Preamble/packages}
\input{Preamble/definitions}
\input{Preamble/metadata}
\input{Preamble/copyright}

\begin{document}

\title{Estimation Error of General Purpose Generative Retrieval}
\title[The Limitations of Constrained Auto-Regressive Decoding in Generative Retrieval]{The Limitations of Constrained Auto-Regressive Decoding\\ in Generative Retrieval}
\title[Constrained Auto-Regressive Decoding Constrains Generative Retrieval]{Constrained Auto-Regressive Decoding\\ Constrains Generative Retrieval}

\input{Preamble/authors}

\input{Sections/00-abstract}

\maketitle

\acresetall

\input{Sections/01-introduction}
\input{Sections/02-related-work}
\input{Sections/03-preliminary}
\input{Sections/04-main-results}
\input{Sections/05-discussion}

\input{Sections/06-synthetic-experiments}
\input{Sections/07-real-experiments}
\input{Sections/08-limitations}
\input{Sections/09-conclusion}

\input{Sections/10-acknowledgements}

\appendix
\section*{Appendix}
\input{Sections/A-appendix}
\input{Sections/B-appendix}
\input{Sections/C-appendix}

\newpage
\bibliographystyle{ACM-Reference-Format}
\bibliography{references}

\end{document}

%% file: Preamble/packages.tex
\usepackage{amsmath}
\usepackage{amsthm}
\usepackage{cancel}
\usepackage{mathtools}
\usepackage{graphicx}
\usepackage{subfigure}
\usepackage{xcolor}
\usepackage{siunitx}
\usepackage[inline]{enumitem}
\usepackage{geometry}
\usepackage{booktabs}
\usepackage{multirow}
\usepackage{acronym}
\usepackage[skip=2pt]{caption}

%% file: Preamble/definitions.tex
\AtBeginDocument{%
  }

\newcommand{\rv}[1]{\boldsymbol{#1}} % random variable
\newcommand{\mc}[1]{\mathcal{#1}}
\newcommand{\ms}[1]{\mathsf{#1}}
\newcommand{\mb}[1]{\mathbb{#1}}
\newcommand{\rd}{\rv{d}} % random variable d
\newcommand{\red}[1]{\textcolor{red}{#1}}
\newcommand{\KL}[2]{\ms{KL}\left[#1\;\|\;#2\right]} % KL divergence
\newcommand{\E}{\mb{E}} % expectation
\newcommand{\Var}{\mb{V}} % variance
\newcommand{\Ent}{\ms{H}} % entropy
\newcommand{\ind}{\mathbb{I}} % indicator function

\newtheorem{theorem}{Theorem}[section]
\newtheorem{lemma}[theorem]{Lemma}

\newtheorem{remark}[theorem]{Remark}

\newcommand{\headernodot}[1]{\noindent\textbf{#1}}
\newcommand{\header}[1]{\headernodot{#1.}}

\acrodef{ctg}[CTG]{controllable text generation}
\acrodef{mips}[MIPS]{maximum inner product search}
\acrodef{gr}[GR]{generative retrieval}
\acrodef{bpb}[bpb]{bits per byte}
\acrodef{bpt}[bpb]{bits per token}
\acrodef{llm}[LLM]{large language model}
\acrodef{docid}[docID]{document identifier}
\acrodef{kl}[KL]{Kullback–Leibler}
\acrodef{ndcg}[nDCG]{Normalized Discounted Cumulative Gain}
\acrodef{mrr}[MRR]{Mean Reciprocal Rank}

%% file: Preamble/metadata.tex
\ccsdesc[500]{Information systems~Retrieval models and ranking}

\keywords{Generative retrieval; Constrained decoding; Beam search}

%% file: Preamble/copyright.tex
\copyrightyear{2025}
\acmYear{2025}
\setcopyright{acmlicensed}
\acmConference[SIGIR '25]{Proceedings of the 48th International ACM SIGIR Conference on Research and Development in Information Retrieval}{July 13--18, 2025}{Padua, Italy}
\acmBooktitle{Proceedings of the 48th International ACM SIGIR Conference on Research and Development in Information Retrieval (SIGIR '25), July 13--18, 2025, Padua, Italy}
\acmDOI{10.1145/3726302.3729934}
\acmISBN{979-8-4007-1592-1/2025/07}

%% file: Preamble/authors.tex
\author{Shiguang Wu}
\orcid{0000-0002-4597-5851}
\affiliation{%
  \institution{Shandong University}
  % \department{School of Computer Science and Technology}
  \streetaddress{72 Binhai Road, Jimo}
  \city{Qingdao}
  \country{China}
  \postcode{266237}
}
\email{shiguang.wu@mail.sdu.edu.cn}

\author{Zhaochun Ren}
\orcid{0000-0002-9076-6565}
\authornote{Corresponding authors.}
\affiliation{%
  \institution{Leiden University}
  \streetaddress{Snellius, Niels Bohrweg 1}
  \city{Leiden}
  \country{The Netherlands}
  \postcode{2333 CA}
}
\email{z.ren@liacs.leidenuniv.nl}

\author{Xin Xin}
\orcid{0000-0003-4703-7356}
\affiliation{%
  \institution{Shandong University}
  % \department{School of Computer Science and Technology}
  \streetaddress{72 Binhai Road, Jimo}
  \city{Qingdao}
  \country{China}
  \postcode{266237}
}
\email{xinxin@sdu.edu.cn}

\author{Jiyuan Yang}
\orcid{0000-0003-2700-5533}
\affiliation{%
  \institution{Shandong University}
  % \department{School of Computer Science and Technology}
  \streetaddress{72 Binhai Road, Jimo}
  \city{Qingdao}
  \country{China}
  \postcode{266237}
}
\email{jiyuan.yang@mail.sdu.edu.cn}

\author{Mengqi Zhang}
\orcid{0000-0001-6831-0740}
\affiliation{%
  \institution{Shandong University}
  % \department{School of Computer Science and Technology}
  \streetaddress{72 Binhai Road, Jimo}
  \city{Qingdao}
  \country{China}
  \postcode{266237}
}
\email{mengqi.zhang@sdu.edu.cn}

\author{Zhumin Chen}
\orcid{0000-0003-4592-4074}
\affiliation{%
  \institution{Shandong University}
  % \department{School of Computer Science and Technology}
  \streetaddress{72 Binhai Road, Jimo}
  \city{Qingdao}
  \country{China}
  \postcode{266237}
}
\email{chenzhumin@sdu.edu.cn}

\author{Maarten de Rijke}
\orcid{0000-0002-1086-0202}
\affiliation{%
  \institution{University of Amsterdam}
  \streetaddress{Science Park 900}
  \city{Amsterdam}
  \country{The Netherlands}
  \postcode{1098 XH}
}
\email{m.derijke@uva.nl}

\author{Pengjie Ren}
\orcid{0000-0003-2964-6422}
\authornotemark[1]
\affiliation{%
  \institution{Shandong University}
  % \department{School of Computer Science and Technology}
  \streetaddress{72 Binhai Road, Jimo}
  \city{Qingdao}
  \country{China}
  \postcode{266237}
}
\email{jay.ren@outlook.com}

\renewcommand{\shortauthors}{Shiguang Wu et al.}

%% file: Sections/00-abstract.tex
\begin{abstract}
\Acf{gr} seeks to replace traditional search index data structures with a single large-scale neural network, offering the potential for improved efficiency and seamless integration with generative large language models.
As an end-to-end paradigm, \ac{gr} adopts a learned differentiable search index to conduct retrieval by directly generating document identifiers through corpus-specific constrained decoding.
The generalization capabilities of \acl{gr} on out-of-distribution corpora have gathered significant attention.
Recent advances primarily focus on the problems arising from training strategies, and addressing them through various learning techniques.
However, the fundamental challenges of generalization arising from constrained auto-regressive decoding still remain unexplored and systematically understudied.

In this paper, we examine the inherent limitations of constrained auto-regressive generation from two essential perspectives: \emph{constraints} and \emph{beam search}.
We begin with the Bayes-optimal setting where the \ac{gr} model exactly captures the underlying relevance distribution of all possible documents. 
Then we apply the model to specific corpora by simply adding corpus-specific constraints.
Our main findings are two-fold:
\begin{enumerate*}[label=(\roman*)]
\item For the effect of constraints, we derive a lower bound of the error, in terms of the KL divergence between the ground-truth and the model-predicted step-wise marginal distributions.
This error arises due to the \emph{unawareness} of future constraints during generation and is shown to depend on the average Simpson diversity index of the relevance distribution.
\item For the beam search algorithm used during generation, we reveal that the usage of marginal distributions may not be an ideal approach.
Specifically, we prove that for sparse relevance distributions, beam search can achieve perfect top-$1$ precision but suffer from poor top-$k$ recall performance.
\end{enumerate*}
To support our theoretical findings, we conduct experiments on synthetic and real-world datasets, validating the existence of the error from adding constraints and the recall performance drop due to beam search.
This paper aims to improve our theoretical understanding of the generalization capabilities of the auto-regressive decoding retrieval paradigm, laying a foundation for its limitations and inspiring future advancements toward more robust and generalizable \acl{gr}.
\end{abstract}

%% file: Sections/01-introduction.tex
\section{Introduction}\label{sec:intro}

The advent of generative models has catalyzed the emergence of \ac{gr} as a new paradigm in information retrieval.
\Ac{gr} provides a potential way to replace the conventional index structure, such as inverted index and vector-based index, with a single large-scale neural network~\cite{metzlerRethinkingSearch2021}.
By integrating the retrieve-then-rank pipeline into an end-to-end framework, \ac{gr} offers the promise of enhanced efficiency. 
Typically, \ac{gr} adopts auto-regressive generative models, e.g., BART~\citep{bart} and T5~\citep{t52020}, trained to generate \acp{docid} given a query.

% problem statement: generalization + constrained beam search
\header{Generalization in neural information retrieval}
% background
% our setting
Generalization is a key problem in neural information retrieval models~\citep{xuGeneralizationAbilityDR2023,yuZeroShotDR2022,gaoZeroShotDR2023,thakur2021beir,xuGeneralizationAbilityDR2023,karpukhinDensePassageRetrieval2020,wangUnsupervisedDomainAdaptationDR2022}.
Similar to the scaling principle in \acp{llm}~\citep{gpt}, many dense retrieval approaches consider training on a high-resource dataset and then evaluated on different domains~\citep{thakur2021beir,karpukhinDensePassageRetrieval2020,wangUnsupervisedDomainAdaptationDR2022,renExaminationZeroShotDR2023,niGeneralizableRetrievers2022}.
In particular, GTR~\citep{niGeneralizableRetrievers2022} demonstrates significant improvements in out-of-domain performance by successfully scaling up the model size and the training corpus.
%Although it was initially proposed for building domain experts~\citep{metzlerRethinkingSearch2021} to integrate with the generative \acp{llm}, 
Unlike dense retrieval, \ac{gr}, as a retrieval paradigm itself, is much more concerned with the generalization abilities to unseen corpora~\citep{chenCorpusBrainPretrain2022,sunLearningTokenizeGenerative2023,askariFewshotIndexing2024,liCorpusLM2024}.

\header{Bayes-optimal \acl{gr}}
Given the success of generative \acp{llm}, \ac{gr} is expected to capture the universal relevance distribution when trained on sufficiently large retrieval datasets.
%with the same architecture and training methodology
While extensive studies have already proposed effective training strategies at large scale~\citep{zengPlanningAheadGenerative2024,zengScalableEffectiveGenerative2023b,liCorpusLM2024,chenCorpusBrainPretrain2022}, the performance of an ideal, \emph{viz.} a Bayes-optimal \ac{gr} model, on unseen corpora has not been systematically studied. 
In this paper, we study an amortized Bayes-optimal auto-regressive \ac{gr} model that fully encapsulates the underlying relevance distribution over the complete corpus containing all possible documents.
%given any query.
On top of this, we apply the constrained auto-regressive generation process to provide a valid \ac{docid} in any given downstream corpus which is a subset of the complete one.

\header{Constrained beam search and generalization}
The core components of \ac{gr} are the differentiable index and the generation of \ac{docid}~\citep{white2025surveyinformationaccess}.
%As the differentiable index is assumed to be perfectly trained in our setting, 
In this paper, we focus on how the generation process affects the generalization of \ac{gr}.
Most existing \ac{gr} models adopt constrained beam search to generate the top-$k$ \acp{docid} as the default retrieval strategy~\citep{tayTransformerMemoryDifferentiable2022a,wangNeuralCorpusIndexer2023,tang2023semantic,sunLearningTokenizeGenerative2023}.
However, \citet{zengPlanningAheadGenerative2024} point out the pitfalls of this strategy, i.e., a greedy local search algorithm is likely to prune false negative \acp{docid} and may thus not be sufficient for developing effective \ac{gr} models.
While these effects are typically entangled with model errors in relevance prediction, the impact of constrained beam search using a Bayes-optimal model with the correct relevance distribution still remains unclear. 
Within our setting, we address this gap by isolating and analyzing the theoretical root of limitations of constrained auto-regressive beam search.

\header{Main findings}
We study the inherent limitations from two essential perspectives: \emph{constraints} and \emph{beam search}.
% For the effect of constraints, we look at the marginal distribution predicted from the model at each step and compare it with the ground-truth distribution.
\begin{enumerate*}[label=(\roman*)]
\item \textbf{Constraints:} In Section~\ref{sec:constraint}, we derive a lower bound on the error, in terms of the KL divergence between the ground-truth and model-predicted step-wise marginal distributions. This error is influenced by the concentration of the relevance distribution and arises from the model's lack of awareness of future constraints during generation. 
\item \textbf{Beam search:} In Section~\ref{sec:recall}, we show that the usage of marginal distributions may not be suitable. We prove that for some sparse relevance distribution with a \emph{thick} tail, beam search can achieve perfect top-$1$ precision, but poor top-$k$ recall performance.
% We also summarize two directions on improving recall:
% \begin{enumerate*}[label=(\roman*)]
% \item \emph{aggregation}, i.e., aggregating the relevant documents within a narrow branch, and 
% \item \emph{amplification}, i.e., amplifying the scores of relevant documents to be exponentially large.
% \end{enumerate*}
% Both directions are meant to provide a high concentration in either local step-wise marginal distribution or global relevance distribution.
{We are aware of existing techniques on improving recall of \ac{gr} model, and summarize two directions based on our analysis:
\begin{enumerate*}[label=(\roman*)]
\item \emph{aggregation}, i.e., aggregating the relevant documents within a narrow branch, and 
\item \emph{amplification}, i.e., amplifying the scores of relevant documents to be significantly large.
\end{enumerate*}
Both directions are essentially meant to provide a high concentration prior on the relevance distribution.
We show that although they are empirically beneficial for small scale corpus, difficulty will arise if we aim to construct a Bayes-optimal \ac{gr}.
}
\end{enumerate*}
% The above results are asymptotic with respect to the vocabulary size of the \ac{gr} model.
% We examine the lower bound of the size required for a sufficiently expressive model, and show that it needs to be exponentially large in terms of the ratio of average raw document length and \ac{docid} length.
Our results provide theoretical grounding of constrained beam search in \ac{gr}, and shows the inherent limitations of this retrieval strategy towards a reliable Bayes-optimal \ac{gr} model.
{
% and the difficulty of directly applying some existing strategies.
Our results also imply the importance of balancing concentration during model training for mitigating this problem.}
% \todo{directions to address this problem or mitigating the effect of decoding?}

\header{Contributions} Our main contributions are as follows: 
\begin{enumerate*}[label=(\roman*)]
    \item We provide theoretical results concerning Bayes-optimal \ac{gr} to determine how constrained beam search affects the generalization.
    \item We decompose the negative effect from two angles, constraints and beam search, and identify a trade-off factor, i.e., the concentration of relevance distribution.
    \item Our theoretical results are verified by experiments on synthetic and real-world datasets.
\end{enumerate*}

%% file: Sections/02-related-work.tex
\section{Related work}\label{sec:related}

\headernodot{\Acf{gr}} is an emerging direction in neural information retrieval, exploring the possibility of replacing traditional index structures in retrieval systems with a single large-scale neural networks~\citep{liSurveyGenerativeIR2024,white2025surveyinformationaccess}.
It leverages generative models to directly generate the relevant \acp{docid} given a query.
This paradigm originated with~\citet{metzlerRethinkingSearch2021,decaoAutoregressiveEntityRetrieval2020} and has garnered considerable attention~\cite{sunLearningTokenizeGenerative2023,wangNeuralCorpusIndexer2023,liLearningRankGenerative2023,Zhuang2022BridgingTG,Zhang2023TermSetsCB,yangAutoSearchIndexer2023,tang2023semantic,tang2024generative,wuGenerativeRetrievalMultiVector2024,seal2022,tayTransformerMemoryDifferentiable2022a,dynamic-retriever2023,nguyen-2023-generative,zengScalableEffectiveGenerative2023b} in the information retrieval community.

%\header{Generalization in \acl{gr}}
%Although it was initially proposed for building domain experts~\citep{metzlerRethinkingSearch2021}, \ac{gr}, as a retrieval system itself, is much concerned with its generalization ability to out-of-distribution corpora~\citep{sunLearningTokenizeGenerative2023,askariFewshotIndexing2024,cont-learning-gr2023cikm,liu2024robustnessgenerative,liuRobustnessGenerativeRetrieval2023,liSurveyGenerativeIR2024,leeNonparametricDecodingGenerative2023}.
Generalization remains a challenge for \ac{gr}, especially when applied to out-of-distribution corpora~\citep{sunLearningTokenizeGenerative2023,askariFewshotIndexing2024,cont-learning-gr2023cikm,liu2024robustnessgenerative,liuRobustnessGenerativeRetrieval2023,liSurveyGenerativeIR2024,leeNonparametricDecodingGenerative2023}. 
Previous research attributes this challenges to limited model capacity~\citep{leeNonparametricDecodingGenerative2023,yuan2024generative-memory-burden}, lack of learning in the docID construction~\citep{sunLearningTokenizeGenerative2023,yangAutoSearchIndexer2023,Zhang2023TermSetsCB}, and difficulties in learning semantic representations~\citep{tang2023semantic,wangNOVOLearnableInterpretable2023}.
In contrast, our work focuses on the constrained auto-regressive decoding strategy widely applied in \ac{gr}, which is crucial for adapting \ac{gr} models to new corpora dynamically.
Our setting aligns closely to the few-shot indexing approach~\citep{askariFewshotIndexing2024}, where a pre-trained \ac{llm} generates \acp{docid} solely based on its pre-trained knowledge and generalization capabilities, without additional training.
We treat their method as a conceptual blueprint for a fully generalizable \ac{gr} system and aim to investigate the inference stage under this setting.
% askariFewshotIndexing2024
% A highly related topic is the updatable \ac{gr}~\citep{kishoreIncDSI2023,mehtaDSIpp2023,cont-learning-gr2023cikm,guoContinualGenerative2024}
% large scale gr
% unified gr and nlp generation models

Updatable \acl{gr} is another critical task on dynamic corpora.
The primary challenges indicate the cost of updating the model with new documents and the catastrophic forgetting problem~\citep{kishoreIncDSI2023,mehtaDSIpp2023,cont-learning-gr2023cikm,guoContinualGenerative2024}.
Previous efforts have concentrated on developing efficient continual learning strategies by fixing the indexing construction procedure. 
We consider an idealized scenario where the model has full knowledge of all possible documents and focuses solely on generating relevant \acp{docid} on dynamic corpus.

\header{Constrained decoding}
Constrained decoding has been widely studied for guiding machine learning models to produce outputs that satisfy specific conditions~\citep{ahmedNeuroSymbolicLearning2023,mustafaStrcutredOutputPrediction2021}.
Instead of learning to satisfy through training, constraints are more often preferable only in the inference time due to the flexibility and efficiency.
\citet{nishinoGeneralizationAnalysisLearning2022a,nishinoUnderstandingCV2025} demonstrate the preservation of relative errors of certain loss functions in realizable setting.
We instead provide a failure case via establishing the existence of a lower-bound error for auto-regressive models operating under step-wise inference-time constraints.
Recent work on \ac{ctg} in \acp{llm}~\citep[see, e.g.,][]{zhangSurveyControllableText2023} also explores imposing constraints during inference without updating the underlying model~\citep{mireshghallahControllableTextGeneration2022,mudgalControlledDecoding2025,kimCriticGuidedDecoding2023,chakrabortyPrincipledDecodingLLM2024}.
However, many of these approaches do not focus on strictly enforcing constraint satisfaction. A few studies~\citep{kimGuaranteedGenerationLarge2024,honghuaLogicalControl2024,zhangTractableControlAutoregressive2023} propose methods to produce outputs that strictly adhere to constraints, mainly hard keywords inclusion constraints, using tractable probabilistic models or policy gradient techniques. 
Our work differs by focusing on a specific corpus-level constraint, i.e., the set of valid \acp{docid} is sampled from the complete corpus, a problem unique to retrieval tasks.

\headernodot{Beam search} is a widely used heuristic algorithm for decoding structured predictors and has been applied as a non-\ac{mips} setup for large-scale retrieval systems with explicit tree structures~\citep{liTreeIndexDenseRetrieval2023,zhuTreeRecsys2018,zhuoOptimalTreeModels2020,zhuJointTreeIndexRecsys2019}.
Beam search is known to have a performance deterioration, and only few works provided theoretical insights into this issue.
As far as we know, only \citet{zhuoOptimalTreeModels2020} demonstrate a training-test discrepancy in tree-structured models using binary cross-entropy loss. They showed that pseudo-labeling during training does not guarantee that beam search will get the most relevant targets.
In our work, we analyze the marginal distribution of an auto-regressive distribution and provide a theoretical result on the top-$1$ and top-$k$ performance under sparse relevance situations.
\citet{zhuoOptimalTreeModels2020} also provide a Bayes-optimal tree structure, which is often called max-heap assumption~\citep{liTreeIndexDenseRetrieval2023}, and we will discuss the difficulty of enforcing this assumption in our setting in Section~\ref{sub:solution}.
In \ac{gr}, some work have reached the same conclusion that beam search is not sufficient for retrieval as it is likely to prune the relevant \acp{docid} and the model is not able to recover from this~\citep{zengPlanningAheadGenerative2024,liCorpusLM2024,liUnigen2024}.
They propose to use a hybrid retrieval strategy to help bypassing this problem.
We instead focus on understanding the root cause of this problem, i.e., the usage of marginal distribution.

% max heap assumption

%% file: Sections/03-preliminary.tex
\section{Preliminaries}
\label{sec:preliminary}

We formulate \ac{gr} and introduce key notations in this section. 

\header{\Acl{gr}}
Following~\citet{tayTransformerMemoryDifferentiable2022a}, we formulate \ac{gr} where the mapping from documents and \acp{docid} is one-to-one function.
%instead of other common variations using n-gram~\citep{seal2022}.
A corpus, denoted as $\mc{D}$, is a set of documents $d$, with each document represented as a sequence of tokens, i.e., $d=(d_1, \ldots, d_m)$, where $m$ is the length. In this paper, we assume all documents have the same length.
A \acl{gr} model $f$, typically implemented using a sequence-to-sequence architecture such as T5~\citep{t52020} or BART~\citep{bart}, generates a ranked list of the most relevant \acp{docid} in $\mc{D}$ for the given query.
The ranking list is computed through beam search during generation. To ensure the model reliably generates valid \acp{docid} from the corpus, a constrained auto-regressive decoding process $g$ together with the beam search is used.

\header{Bayes-optimal \acl{gr}}
We first assume the complete corpus $\mc{D}$ contains all possible documents of length $m$. Any downstream corpus is therefore a subset of $\mc{D}$.
We denote $f$ as the Bayes-optimal \ac{gr} model on $\mc{D}$ which has the ability to predict the ground-truth relevance distribution over $\mc{D}$ given any query.
The Bayes-optimal model $f$ is considered as an ideal and oracle prototype model which helps us understand the behavior of the generation process.
When $f$ is applied to a downstream corpus $\mc{D}^c\subset \mc{D}$, it uses the corresponding constrained decoding process $g^c$ to predict relevant documents in $\mc{D}^c$. 
This induced \ac{gr} model on $\mc{D}^c$ is denoted as $f^c$.
Note that a similar setting has recently been proposed as zero-shot indexing~\citep{askariFewshotIndexing2024}.

\input{Tables/notation}
\header{Notation}
Table~\ref{table:notation} lists the main notation used in the paper.
For an integer $n$, we denote the set $\{1,\ldots,n\}$ by $[n]$.
We use bold face to denote random variables.
Tokens in a document $d$ are integers from $[k]$, making $d\in [k]^m$, and $k$ is the vocabulary size.
We set the complete corpus $\mc{D}=[k]^m$.
We denote the underlying relevance distribution over $\mc{D}$ given a query $q$ as $\Pr(\cdot \mid q)$, where $\rd \sim \Pr(\rd \mid q) =\prod \Pr(\rd_i \mid \rd_{<i},q)$.
For simplicity, we focus on a single query $q$ and omit it in some context.
We use subscripts to indicate a sliced subset or operations at specific steps, i.e., $\Box_{i}$, and $\Box_{\ge i}$, etc.
Particularly, we refer to $\mc{D}_{\ge i}$ or $\mc{D}^c_{\ge i}$ a branch with root $d_i$, under some prefix $d_{<i}$.

\header{Downstream corpus and constraints}
We construct the downstream corpus $\mc{D}^c$ by sampling.
For simplicity, each document is sampled with an equal probability $p$.
In practice, the sampling is agnostic to the future user queries, and thus independent with $\Pr(\cdot)$.
We use $C$ to be the event that $\rd$ is in $\mc{D}^c$, and $\Pr(\cdot\mid C)$ is the distribution under the corpus $\mc{D}^c$.
$C_{i}$ means the $i$-th token of document $\rd$ is valid with respect to the downstream corpus constraints, given some context-clear prefix tokens $d_{<i}$, i.e., $d_i$ is valid if it appears in $\mc{D}^c$ for the prefix $d_{<i}$.
We use ``constraints'', and ``downstream corpus'' to represent the result of sampling interchangeably.
The constrained generation process $g^c_i$ is applied at the $i$-th step of $f$.
It first zeros out the invalid tokens and then re-normalizes the remaining probabilities.

%% file: Tables/notation.tex
\begin{table}[t]
\caption{Glossary.}
\label{table:notation}
\begin{tabular}{ll}
    \toprule
    \textbf{Symbol} & \multicolumn{1}{c}{\textbf{Description}}
    \\
    \midrule
    $k,\,m$         & the vocabulary size and document length \\
    $\mc{D}$        & the complete corpus, i.e., $[k]^m$ \\
    $\mc{D}^c,\,C$  & downstream corpus and constraints \\
    $f,\,f^c$       & Bayes-optimal and induced downstream \ac{gr} \\
    $\Pr(\cdot \mid q)$    & relevance distribution given query $q$ \\
    $d=(d_1, \ldots,d_m)$   & document in $\mc{D}$ \\
    \bottomrule
\end{tabular}    
\end{table}

%% file: Sections/04-main-results.tex
\section{Theoretical analysis}
\label{sec:main_results}
We investigate the inherent limitations of \ac{gr} arising from constraints and beam search individually.
% We explicitly disentangle these two components.
Our analysis disentangles these factors to isolate their individual effects. 
For constraints, we analyze its effect on the marginal distribution at each generation step.
For beam search, we analyze how it independently degrades recall performance on the complete corpus, disregarding constraints.
These results are asymptotic with respect to the vocabulary size $k$.
In Appendix~\ref{sec:complexity}, we examine the required magnitude of $k$ for a sufficiently expressive model and show that it needs to be exponentially large in terms of the ratio of raw document length and \ac{docid} length.

\subsection{Constraints cause marginal distribution mismatch}
\label{sec:constraint}

We begin by studying the effects of applying constraints to the model performance.
We first identify the factor that causes the error during the generation.
Then we quantitatively analyze the magnitude of this error in
\begin{enumerate*}[label=(\roman*)]
    \item uniform, and
    \item general
\end{enumerate*}
relevance distribution given some query $q$.
We consider the first generation step without loss of generality.
Other cases can be reduced to it by adjusting the document length $m$ or vocabulary size $k$.

\header{Unawareness of future constraints}\label{sub:est_error}
% Since $f^c$ auto-regressively generates the document tokens using the factorized distribution $\Pr(\rd_i \mid \rd_{<i},C)$, we will examine in detail on how \ac{gr} model miscalculates it.
Since $f^c$ is \emph{unaware} of the constraints in the future steps, there may exist biases between the distributions of complete corpus $\mc{D}$ and downstream corpus $\mc{D}^c$.
Specifically, after applying constrained decoding $g^c$, $f^c(q)$ returns
\begin{align}
    \Pr(\rd_1 \mid q,C_1)
    & {} =g^c[\Pr(\rd_1 \mid q)] \\
    & {} \propto \ind[d_1\in \mc{D}^c_1]\sum_{\forall d_{>1}\in\red{\mc{D}_{>1}}} \Pr(d_1d_{>1} \mid q),
    \label{eq:gr_marginal}
\end{align}
where the $C_1$ constraint is satisfied through $g^c$, and $\ind[d_1\in\mc{D}^c_1]$ means $d_1$ is a valid first token in the downstream corpus.
In the contrary, the ground-truth marginal distribution only sum over the documents in $\mc{D}^c$, so we have
\begin{align}
    \Pr(\rd_1 \mid q,C_1,\red{C_{>1}})
    & {} = g^c[\Pr(\rd_1 \mid q,\red{C_{>1}})] \\
    & {} \propto  \ind[d_1\in \mc{D}^c_1]\sum_{\forall d_{>1}\in\red{\mc{D}^c_{>1}}} \Pr(\rd_1d_{>1} \mid q),
    \label{eq:real_marginal}
\end{align}
where $\Pr(\rd_1\mid q,C_1,C_{>1})=\Pr(\rd_1\mid q,C)$.
Here we use the \red{red} mark to highlight the differences from Eq.~\ref{eq:gr_marginal}.
Note that this gap would not arise if the downstream corpus were preset, with both training and inference performed on it, as the model would learn $\Pr(\rd\mid q, C)$ directly.
We then analyze the \ac{kl} divergence as follows,\footnote{The reason we adopt the \ac{kl} divergence is that it is not only part of the training loss, i.e., empirical cross-entropy, but also related to the ranking performance.
Several publications have showed that the cross-entropy is a bound of several commonly used metrics, e.g., \ac{ndcg} and \ac{mrr}~\citep{bruchAnalysisCEforl2r2019,yang2024psl} for binary relevance score.
\ac{kl} divergence and cross-entropy only differ by the entropy.}
\begin{align}
       & {} \KL{\Pr(\cdot \mid C)}{\Pr(\cdot \mid C_1)} \\
    =\ & \E_{\rd_1\sim \Pr(\cdot \mid C)}\left[\log \frac{\Pr(\rd_1 \mid C)}{\Pr(\rd_1 \mid C_1)}\right] \\
    =\ & \E\left[\log \Pr(C_{>1} \mid \rd_1)\right]-\log \Pr(C_{>1} \mid C_1).
    \label{eq:kl_error}
\end{align}
In Eq.~\ref{eq:kl_error}, we see that the \ac{kl} divergence is the gap between the average proportion of constraints locally within each branch $\mc{D}^c_{\ge 1}$ and the global average constraints.
This indicates that the variation across branches may contribute the mismatch.
Recall that we construct our constraints via an i.i.d.\ sampling, and we have the expectation of the constrained marginal distribution $\Pr(d_1\mid q,C_1)$ as follows:
\begin{equation}
    \E_C[\Pr(d_1\mid q,C)]\propto \E[\ind[\rd\in\mc{D}^c]]\Pr(d_1\mid q) \propto \Pr(d_1\mid q).
\end{equation}
Therefore, the downstream corpus will follow the same distribution as the complete corpus on average.
However, as we will see in the next, the gap is not concentrated at zero with high probability in terms of \ac{kl} divergence.
% \todo{\ac{kl} div. is related to ranking.}

\header{Uniform relevance distribution}\label{sub:eg_uniform}
We first give an example case when the relevance distribution $\Pr(\rd\mid q)$ is a uniform distribution.
If the size of the downstream corpus is $k^r\ll k^m=|\mc{D}|$, we derive an asymptotic lower bound of the error for large $k$ as follows
\begin{equation}
    \KL{\Pr(\cdot \mid C)}{\Pr(\cdot \mid C_1)}\gtrsim \frac{0.05}{k^{r-1}}.
    \label{eq:kl_lower_uniform}
\end{equation}
In particular, we have a constant error $0.05$ if the size is $O(k)$.
For detailed illustration of the proof, see Theorem~\ref{the:kl_uniform}.
Here is the intuition:
\Ac{gr} model will predict a uniform distribution over the valid first tokens, but the ground-truth should be proportional to the number of valid documents in each valid branch.
Due to the variance of the sampling, the ground-truth distribution will not be exactly uniform.

\header{General relevance distribution}\label{sub:eg_arbitrary}
Following the same idea of the uniform case, we further give the result for general relevance distributions.
The key is the Simpson diversity index~\citep{simpson_diversity_index}, which is used for measuring the degree of concentration.
We introduce the average Simpson diversity index.
It is computed as the squared expectation of root sum of squared probabilities $\Pr(d\mid d_1)$, i.e.,
\begin{equation}
    \E_{d_1}^2\left[\sqrt{\sum_d \Pr(d\mid d_1)^2}\right].
\end{equation}
Recall that each document is selected with probability $p$, and let $A$ be the average Simpson diversity index, we have an asymptotic lower bound of the \ac{kl} divergence for large $k$ in Eq.~\ref{eq:kl_lower_non_uniform}.
\begin{equation}
    \KL{\Pr(\cdot \mid C)}{\Pr(\cdot \mid C_i)}\gtrsim \frac{0.05A}{p}.
    \label{eq:kl_lower_non_uniform}
\end{equation}
For detailed illustration of the proof, see Theorem~\ref{the:kl_arbitrary}.
We also show that the lower bound reaches its minimum in the uniform relevance distribution case. 
Note that we have $A\ge \frac{k}{|\mc{D}|}$, where the equality is obtained when $\Pr(\cdot)$ is uniform.
Let the selected corpus size be $k^r\approx p|\mc{D}|$, we have the same error shown in Eq.~\ref{eq:kl_lower_uniform}.
\begin{equation}
    \frac{0.05A}{p}\approx \frac{0.05k}{|\mc{D}|p}=\frac{0.05}{k^{r-1}}.
\end{equation}
Note that for concentrated distribution, the Simpson diversity index considerably exceeds $p$, resulting in a corresponding larger error.

In conclusion, we infer a lower bound of the \ac{kl} divergence between the predicted and ground-truth marginal distribution.
The bound is proportional to the degree of concentration of the underlying relevance distribution.

% =============================================================

\subsection{The impact of beam search on recall}\label{sec:recall}
Beam search often fails to retain the correct top-$k$ prefix candidates, resulting in the exclusion of relevant documents during generation~\citep{zengPlanningAheadGenerative2024}. 
We attribute this limitation to the usage of conditional decomposition of the joint distribution over the corpus.

To formalize this analysis, we model the decomposition as a tree where:
\begin{enumerate*}[label=(\roman*)]
    \item each node at layer $i$ represents a token $d_i$ generated under a specific prefix $d_{<i}$,
    \item each sub-tree represents possible continuations of a prefix, and
    \item the \emph{value} of node $d_i$ with prefix $d_{<i}$, denoted $V(d_i\mid d_{<i})$, equals $\sum_{d_{>i}}\Pr(d_{\ge i}\mid d_{<i})$, i.e., the marginal probability $\Pr(d_i\mid d_{<i})$.
\end{enumerate*}
The property of this structure is that node values represent marginal probabilities aggregated over all possible future paths.
However, the objective of the retrieval model requires identifying specific document with maximal joint probability $\Pr(d)$.
This creates a fundamental mismatch: nodes with high values $V(d_i\mid d_{<i})$ may not contain the hightest-joint-probability documents in their sub-trees.

\header{Non-relevant branches overtaking relevant ones}
\label{sub:overtaken}
We setup a scenario with a sparse relevance distribution to elucidate this issue.
We present how branches containing relevant documents can be overtaken by non-relevant peer branches during generation.
At the first generation step, the model have to choose within $k$ nodes, each associated with a sub-tree of $k^{m-1}$ documents (leaves). 
We assign a logit uniformly sampled from $[-1, 1]$ to each document in the corpus. 
A subset of $\lambda k \ll |\mc{D}|$ documents is randomly selected as relevant, and each is assigned a logit within $O(\log k \pm \log\log k)$. 
The exponential of each logit is the final relevance score of the document, i.e., the score of the corresponding root-to-leaf path.

We prove that the recall of the top-$\lambda k$ valued branches is upper-bounded by $0.5 + o(1)$ with high probability. 
This indicates that many relevant branches are excluded from the highest-valued branches. 
However, the top-$1$ branch is highly likely to contain the most relevant documents. 
A detailed formal statement and proof of this result is provided in Theorem~\ref{the:recall} in Appendix~\ref{app:recall}.
The thickness of the tail distribution or the sharpness of relevant branches determines the probability of overtaking.
If the scores for relevant branches are sufficiently high, this issue becomes less pronounced.

In summary, \ac{gr} models relying on the sum of sub-tree values for ranking branches struggle to achieve high recall performance while maintaining top-$1$ precision. 

\subsection{Concentration as a trade-off factor}
We find a common dependence on the concentration of relevance distribution from both components, i.e., the Simpson diversity index, and the thickness of the tail distribution.
However, the concentration takes effect in opposite directions.
Less concentrated distribution is more stable against sampling and thus have better alignment with the ground-truth constrained marginal distribution.
In contrast, high concentration reduces the false positive branches and thus improve the recall performance using beam search, but it requires an accurate marginal distribution.

%% file: Sections/05-discussion.tex
\section{Difficulty of concentration in Bayes-optimal \ac{gr}}\label{sub:solution}

While increasing the degree of concentration is a promising approach to improve recall performance on the complete corpus---if the effects of constraints are ignored---achieving this in the context of Bayes-optimal \ac{gr} is still challenging. 
In this section, we show details about the challenge in the context of Bayes-optimal \ac{gr} in two directions:

\begin{enumerate}[leftmargin=*]
    \item \textbf{Aggregation:} Aggregating relevant documents within a narrow branch to ensure that its value significantly exceeds that of others throughout the generation process.
    \item \textbf{Amplification:} Amplifying the scores of relevant documents relative to non-relevant ones, so that the branch value serves as a reliable sketch for branch quality.
\end{enumerate}
While these strategies have shown their effectiveness~\citep{liTreeIndexDenseRetrieval2023,zengScalableEffectiveGenerative2023b} on small scale, in-domain datasets, they are much more challenging to implement in Bayes-optimal \ac{gr} models.
Specifically, as we will discuss, achieving this concentrated structure requires additional computational and data resources beyond those needed to obtain the Bayes-optimal model itself.
Detailed discussions are shown in the following.

\subsection{Aggregation introduces redundancy}
Although the corpus size $|\mc{D}|=k^m$ is sufficiently large, the model does not necessarily fully utilize the entire code space.
Here we study the entropy of the corpus distribution $\Pr(\cdot)$, which is defined as $\Ent(\rv{d}) = -\sum_{d\in\mc{D}} \Pr(d)\log\Pr(d)$.
It can be decomposed into the entropy of the marginal distribution at each step, i.e., $\Ent(\rv{d}) = \sum_{i=1}^m \Ent(\rv{d}_i\mid \rv{d}_{<i})$.
When the distribution is uniform, the entropy is maximized, i.e., $\Ent(\rv{d}) = \log |\mc{D}| = m\log k$.
If we introduce external prior knowledge to aggregate the relevant documents, the relevant branches will stand out at a very early stage of generation, and the entropy at that layer $\Ent(\rv{d}_i\mid \rv{d}_{<i})$ will be low and even approaching zero.
The more aggregated the relevant documents are, the lower the entropy becomes.
% Therefore, for the same code space, more aggregated structure will store less information.
In other words, for the same corpus size, the model will waste more code space on redundant structures.
In practice, this approach for concentration is often realized by conducting hierarchical clustering on the corpus, which has been shown to be effective in small scale retrieval tasks.
The trade-off between concentration and redundancy can be ignored to some extent when the corpus is small.
However, if we would like to build a Bayes-optimal \ac{gr} model, learning a sufficiently large code space is already expensive, see Appendix~\ref{sec:complexity}.
Effective concentration will introduce more redundancy, which is computationally inefficient.

\subsection{Amplification requires high-quality data}
We treat the amplification strategy as an approximation of the max-heap structure discussed in \citet{liTreeIndexDenseRetrieval2023}.
In this structure, the value of each node is the maximum value of its children instead of the sum, i.e., $V(d_i\mid d_{<i}) = \max_{d_{>i}} \Pr(d_i, d_{>i}\mid d_{<i})$.
One can prove that this structure can achieve perfect recall performance by preserving the relevant documents in the top-$n$ branches~\citep{zhuoOptimalTreeModels2020}.
Note that this structure is no longer a chain decomposition of the joint distribution, and the distribution at each step is different from the original distribution.
The new distribution can still be learned through empirical risk minimization by carefully filtering the training data.
Considering the $i$-th step, the negative log-likelihood loss is as follows:
\begin{equation}
    \mc{L}^i(\theta) = -\sum_{d\in\Tilde{\mc{D}}} \log\Pr(d_i\mid d_{<i};\theta),
\end{equation}
where $\Tilde{\mc{D}}$ is the training set, and $\theta$ is the model parameter to be optimized.
In order to learn the max-heap structure, the predicted distribution $\Pr(d_i\mid d_{<i};\theta)$ should be proportional to $V(d_i\mid d_{<i})$.
Therefore the loss function should filter out the non-maximum successors in the training set, i.e.,
\begin{equation}
\mbox{}\hspace*{-3mm}
    \mc{L}^i(\theta) \!=\! -\!\!\sum_{d\in\Tilde{\mc{D}}} \!\ind[d\in\arg\max \Pr(d_{> i}\mid d_{\le i})]\log\Pr(d_i\mid d_{<i};\theta).
\hspace*{-2mm}\mbox{}
\end{equation}
As we can see, only the most relevant successors are allowed to contribute to the loss function.
It not only needs to throw away large amounts of data but also requires more high-quality data and careful filtering strategies.
Most retrieval models can be seen to be trained on this loss~\citep{liTreeIndexDenseRetrieval2023}, as they are only trained on the labeled relevant documents.
However, the quality and quantity of retrieval data are often limited considering that the amount of data needed for a generative model is generally much larger than that for a discriminative model.

%% file: Sections/06-synthetic-experiments.tex
\section{Synthetic experiments}\label{sec:synthetic-exp}

In this section, we experimentally verify the theoretical results and investigate scenarios beyond the assumed data distributions presented in Section~\ref{sec:constraint} and Section~\ref{sec:recall}.

\subsection{Effects of constraints}

\input{Figures/constraints_synthetic_uniform}

\header{Uniform relevance distribution}  
We begin by simulating the case of a uniform relevance distribution, as discussed in Section~\ref{sub:eg_uniform}.  
Since the lower bound is expressed in terms of the size of the downstream corpus, we vary its size to observe the behavior across different vocabulary sizes.  
We choose a sufficiently large $m$ so that the complete corpus of size $k^m$ is much larger than the downstream corpus.  
For a downstream corpus of size $n$, the sampling probability is given by $p=\frac{n}{k^m}$.  
The simulation results are shown in Figure~\ref{fig:constraints_error_synthetic_uniform}.  

As illustrated, the \ac{kl} divergence for different $k$ values exhibits a consistent decreasing trend.  
For a fixed downstream corpus size, a smaller $k$ results in significantly lower error.  
This is expected, as a smaller $k$ leads to each branch covering a larger number of selected documents, thereby reducing variance.  
In this case, the model-predicted marginal distribution closely follows the actual one.

\input{Figures/constraints_synthetic_general}
\vspace{-3pt}
\header{General relevance distribution}
Next, we simulate the case of general relevance distributions, as described in Section~\ref{sub:eg_arbitrary}.  
For simplicity, we ensure that each branch has the same Simpson diversity index.  
We assign random weights to documents within the range $[1, e^{100}]$, where the assignment probability decreases exponentially with larger weights.  
The concentration of the distribution is controlled by varying the rate of decrease: slower rates result in less concentrated distributions, corresponding to smaller Simpson diversity indices.  
When the Simpson diversity index approaches $1$, our lower bound no longer holds.  

The simulation results are shown in Figure~\ref{fig:constraints_error_synthetic_general}.  
Compared to the uniform distribution case, the \ac{kl} divergence is significantly larger and decreases more slowly as the downstream corpus size increases. 
For data with a Simpson diversity index approaching $1$, the \ac{kl} divergence reaches approximately $6$.
In contrast, the uniform case consistently maintains a very low \ac{kl} divergence.
For example, the lowest Simpson diversity index in the figure is around $1e-06$, which matches the magnitude of the uniform distribution, $\frac{1}{2^{20}} \approx 1e-06$, and the corresponding \ac{kl} divergence is approximately $0.19$, as seen in Figure~\ref{fig:constraints_error_synthetic_uniform}.

\subsection{Effects of beam search}

\input{Figures/recall_synthetic_lambda}

\header{Verification of theoretical results}
We first construct a sparse relevance distribution as outlined in Section~\ref{sub:overtaken}.  
The results are presented in Figure~\ref{fig:recall_synthetic_lambda}.  
We evaluate the recall of relevant branches under varying values of $\lambda$, which controls the number of relevant documents sampled.  
Specifically, we compute how many relevant branches are preserved in the top positions during the first generation step.  
We also examine large $\lambda$ values, which are beyond the scope of Theorem~\ref{the:recall}.  

For small $\lambda$, the recall performance is approximately $0.5$.  
For larger $\lambda$ values, the recall still aligns with the theoretical bounds.  
In all cases, the $\text{precision}@1$ remains consistently perfect.  
Regarding score magnitudes, as each branch has a fixed size of $n=2^{25}$, relevant documents achieve scores around $e^{8.6} \approx 5400$, while non-relevant documents score around $e^1 \approx 2.7$.

\input{Figures/recall_synthetic_temperature}

\header{Effects on different degrees of concentration}
We also investigate the impact of varying the sharpness of the relevance distribution by introducing a temperature parameter $T$.  
For a document with logit $s$, the score is computed as $\exp(s/T)$.  
Lower temperatures increase the gap between scores for relevant and non-relevant documents.  
The results are shown in Figure~\ref{fig:recall_synthetic_temperature}.  

Our findings reveal that recall performance is highly sensitive to temperature, achieving perfection within a narrow range.  
This confirms the advantage of constructing \ac{gr} models that capture concentrated distributions effectively.

\subsection{Summary}
The results from the synthetic data distributions validate the theoretical findings presented in earlier sections.  
Although these synthetic settings are not practical for real-world scenarios, they provide a controlled environment to clearly demonstrate the negative effects of constrained decoding and beam search.

%% file: Figures/constraints_synthetic_uniform.tex
\begin{figure}[t]
    \centering
    \includegraphics[width=0.9\columnwidth]{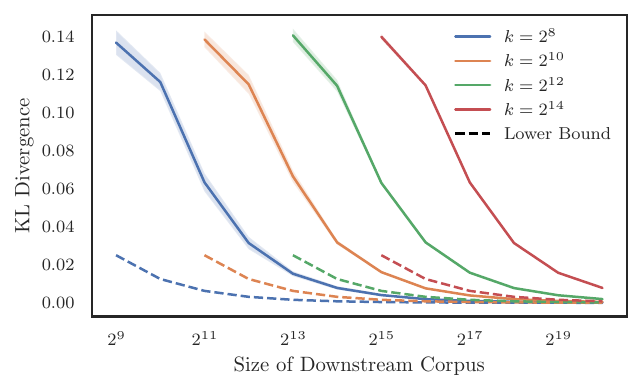}
    \caption{The KL divergence error in the first generation step on synthetic uniform relevance distribution data with uniformly sampled downstream corpus.}
    \label{fig:constraints_error_synthetic_uniform}
\end{figure}

%% file: Figures/constraints_synthetic_general.tex
\begin{figure}[t]
    \centering
    \includegraphics[width=0.9\columnwidth]{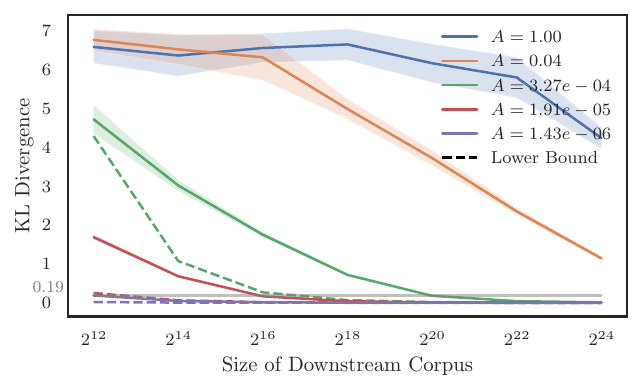}
    \caption{The KL divergence error in the first generation step on several synthetic relevance distribution data with different degrees of concentration.
    The vocabulary size is $2^{10}$, the \ac{docid} length is $3$.
    $A$ is the Simpson diversity index of the relevance distribution.
    As for highly concentrated distributions, e.g., $A=1$, and $A=0.04$, the Lyapunov's condition will no longer hold (see Theorem~\ref{the:kl_arbitrary} for more details), we do not draw their lower bounds.}
    \label{fig:constraints_error_synthetic_general}
\end{figure}

%% file: Figures/recall_synthetic_lambda.tex
\begin{figure}[t]
    \centering
    \includegraphics[width=0.9\columnwidth]{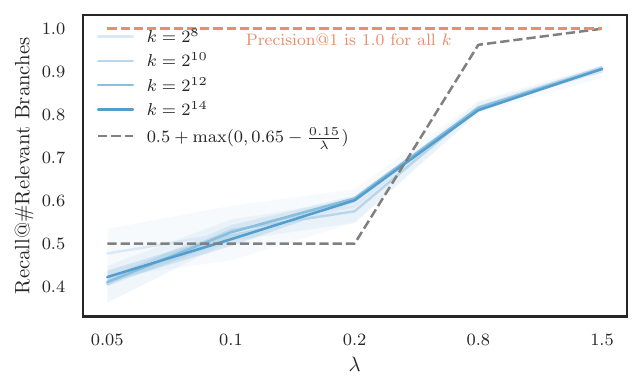}
    \caption{The recall of relevant branches cut off at the total number of relevant branches in the first generation step.
    The synthetic relevance distribution is constructed as Section~\ref{sub:overtaken}.
    The total number of sampled relevant document is $\lambda k$.
    The size of each branch is fixed as $n=2^{25}$.}
    \label{fig:recall_synthetic_lambda}
\end{figure}

%% file: Figures/recall_synthetic_temperature.tex
\begin{figure}[t]
    \centering
    \includegraphics[width=0.9\columnwidth]{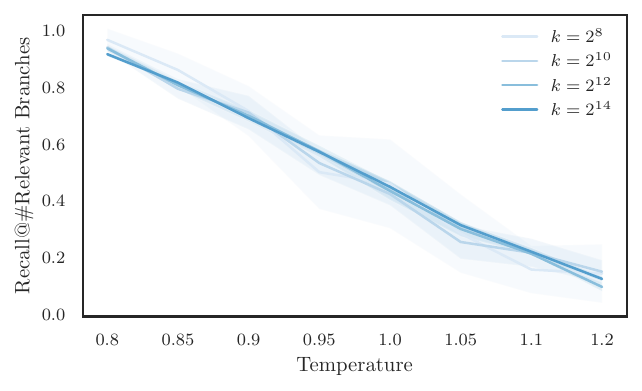}
    \caption{The recall of relevant branches cut off at the total number of relevant branches in the first generation step.
    The synthetic relevance distribution is constructed as Section~\ref{sub:overtaken}.
    The total number of sampled relevant document is $0.05 k$.
    The size of each branch is fixed as $n=2^{25}$.}
    \label{fig:recall_synthetic_temperature}
\end{figure}

%% file: Sections/07-real-experiments.tex
\vspace{-3mm}
\section{Experiments on real-world dataset}\label{sec:real-exp}

\subsection{Experimental setups}
\label{sub:setup}
\textbf{\Acl{docid} design.}
For the \ac{docid} design, we adopt the codebook and semantic ID mapping from~\citet{zengScalableEffectiveGenerative2023b}, which introduces the first effective \ac{gr} model that outperforms conventional retrieval models on the full MS MARCO passage corpus~\citep{bajaj2016ms}. 
The codebook size is set to $256$.

\header{Datasets}
We evaluate our approach using the MS MARCO V1 passage corpus~\citep{bajaj2016ms}, which contains $8.8$ million passages, along with three evaluation datasets:
\begin{enumerate*}[label=(\roman*)]
    \item MS MARCO-dev, consisting of $7$k queries;
    \item TREC DL 2019~\citep{craswell2020overviewtrec2019deep}, with $43$ queries; and
    \item TREC DL 2020~\citep{craswell2021overviewtrec2020deep}, with $54$ queries.
\end{enumerate*}

\input{Figures/constraints_10k}
\header{Relevance distribution}
While recent advance improves \ac{gr} performance~\citep{zengScalableEffectiveGenerative2023b}, there remains a gap compared to state-of-the-art models. 
We use SLIM++~\citep{liSLIM2023} to compute relevance scores due to its strong overall performance on the MS MARCO V1 passage re-ranking leaderboard.\footnote{MS MARCO V1 Passage Regressions are available at \href{https://castorini.github.io/pyserini/2cr/msmarco-v1-passage.html}{Pyserini}.} 
We use the top-$10,000$ scores and extrapolate the remaining scores by taking $1\%$ of the lowest score from the ranked list, adding small Gaussian perturbations.
\input{Figures/constraints_full}
\vspace{-2mm}
\subsection{Effect of constraints}
We examine two sampling strategies: uniform sampling, as discussed in previous sections, and sampling directly based on the relevance distribution.
% \footnote{For sampling using the relevance distribution, we randomly select without replacement until the desired size is reached, implemented via \texttt{numpy.random.choice}.}
In practice, sampling should remain agnostic to future relevance, but this analysis illustrates what occurs if the downstream corpus distribution aligns with the complete corpus.

For uniform sampling, Figures~\ref{fig:constraints_10k.uniform} and~\ref{fig:constraints_full.uniform} depict the \ac{kl} divergence across varying complete corpus sizes. 
We observe that the error decreases rapidly as the size of the downstream corpus increases. 
For small $k$, the Simpson diversity index in each branch is lower due to the higher number of elements, which helps reduce the \ac{kl} divergence as defined in Eq.~\ref{eq:kl_lower_non_uniform}. 

In contrast, for non-uniform sampling, Figures~\ref{fig:constraints_10k.non_uniform} and~\ref{fig:constraints_full.non_uniform} show that the error is substantially higher than in the uniform case. 
Non-uniform sampling introduces greater variance and exacerbates the mismatch between distributions.

\subsection{Effect of beam search}
We evaluate recall performance by examining the retrieval of relevant branches. 
The cutoff is set to the number of relevant documents, which does not exceed the number of relevant branches. 
Results are presented in Table~\ref{table:recall}. 
The recall is relatively low, indicating that branches containing highly relevant documents may not always be retrieved within the top-$k$ branches. 
However, the top-$1$ precision is higher, particularly for the MS MARCO-dev evaluation set.

\input{Tables/recall}

\subsection{Summary}
Our experimental results demonstrate that the theoretical findings reveal the negative effects of constrained beam search under real-world dataset distributions to some extent. 
Since \citet{zengScalableEffectiveGenerative2023b} designed and trained the \ac{gr} model specifically on the MS MARCO passage corpus, the \ac{docid} structure is highly adapted to that query-document training distribution. 
Additionally, the size of the MS MARCO corpus is far from ideal for representing a complete corpus. 
As a result, the experimental findings only provide an approximate reflection of the theoretical results.\footnote{The source code is available at \href{https://github.com/Furyton/constrained-generation-in-gr}{ https://github.com/Furyton/constrained-generation-in-gr}.}

%% file: Figures/constraints_10k.tex
\begin{figure}[t]
  \centering
  \subfigure[Uniform sampling]{\includegraphics[height=0.49\linewidth]{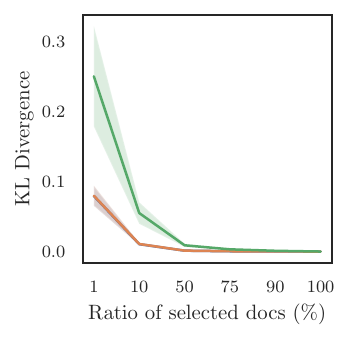}
  \label{fig:constraints_10k.uniform}}
  \hfill
  \subfigure[Non-uniform sampling]{\includegraphics[height=0.49\linewidth]{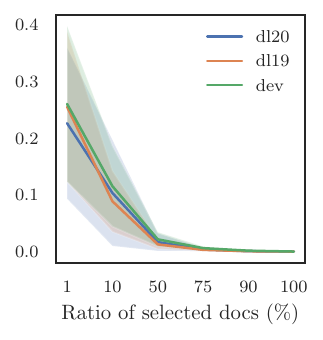}
  \label{fig:constraints_10k.non_uniform}}
  \caption{The KL divergence error in the first generation step on MS MARCO V1 passage corpus.
    A subset of top-$10,000$ rank list from SLIM++~\citep{liSLIM2023} is treated as the complete corpus, and the downstream corpus is (a) uniformly sampled, or (b) sampled with the relevance distribution.
    }
  \label{fig:constraints_10k}
\end{figure}

%% file: Figures/constraints_full.tex
\begin{figure}[t]
  \centering
  \subfigure[Uniform sampling]{\includegraphics[height=0.48\linewidth]{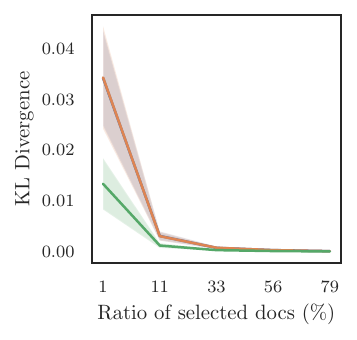}
  \label{fig:constraints_full.uniform}}
  \hfill
  \subfigure[Non-uniform sampling]{\includegraphics[height=0.48\linewidth]{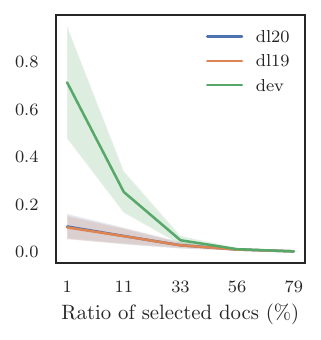}
  \label{fig:constraints_full.non_uniform}}
  \caption{The KL divergence error in the first generation step on MS MARCO V1 passage corpus.
    The relevance distribution is constructed as Section~\ref{sub:setup}.
    The downstream corpus is (a) uniformly sampled, or (b) sampled with the relevance distribution.
    }
  \label{fig:constraints_full}
\end{figure}

%% file: Tables/recall.tex
% mean_recall': np.float64(0.7748458041608566), 'mean_precision': np.float64(0.9074074074074074) @100, dl20
% Mean Recall: 0.6746265321476796, Mean Precision: 0.9053497942386831 @50, dev
% Mean Recall: 0.5372275001930796, Mean Precision: 0.6976744186046512, @50, dl19
% cutoff: 50, tag: dl20, Mean Recall: 0.6331761843257662, Mean Precision: 0.7592592592592593

\begin{table}[ht]
\caption{Recall@50 and Precision@1 of relevant branches ($256$ in total) at the first generation step, which are denoted as ``R@50'' and ``P@1''.
The highest $50$ documents are labeled relevant, and the branches contain these documents are labeled relevant.
The relevance distribution is constructed in Section~\ref{sub:setup}.
}
\label{table:recall}
\begin{tabular}{l cc}
    \toprule
    \textbf{Model} & R@50 & P@1 \\
    \midrule
    TREC DL 19 & 53.7 & 69.8 \\
    TREC DL 20 & 63.3 & 75.9 \\
    MS MARCO-dev & 67.5 & 90.5 \\
    \bottomrule
\end{tabular}
\end{table}

%% file: Sections/08-limitations.tex
\section{Limitations}\label{sec:limitation}
Concerning the theoretical aspects of our work, as we do not have an evidence of what an optimal \ac{gr} model should look like, we fail to provide practical assumptions on the relevance distribution and the structure of \ac{docid}.
Our results are sensitive to the parameters and assumptions, e.g., the sharpness of the relevance distribution for Theorem~\ref{the:recall}, and may not accurately reflect practical real-world situations.
Besides, we have not studied how the two factors affect each other when using constrained beam search.

Concerning the experimental aspects of our work, we only use the \ac{docid} from \citet{zengScalableEffectiveGenerative2023b} and the MS MARCO V1 passage corpus.

%% file: Sections/09-conclusion.tex
\vspace{-3mm}
\section{Conclusion}\label{sec:conclusion}
In this paper, we have provided theoretical results on the effect of constrained beam search for a Bayes optimal \ac{gr} model.
We have considered two separate aspects, constraints and beam search, and examine the root cause of the negative effect on generalization.
Both aspects are intrinsically connected to the degree of concentration of the relevance distribution across the complete corpus.
When it is more concentrated, the model achieves decent recall performance, provided the marginal distribution aligns closely with the actual one. 
However, applying downstream corpus constraints increases this marginal distribution gap at the same time.
To validate our theoretical findings, we have conducted experiments on synthetic and real-world dataset and have shown the case beyond the assumed conditions in the theorem.
% Our results not only predict the errors in synthetic environment but also generalize to real-world complex distributions.
Overall, we give a systematical investigation from both theory and experiments to the limitation of constrained decoding on retrieval generalization.

Based on our findings, practitioners in the field may consider balancing the degree of concentration when designing and training \ac{gr} model, and using post-calibration to fix the errors when using the model on a downstream corpus.
Other forms of decoding strategies beyond constrained beam search are also suggested.

As to future work, we will continue to study how these results can be used for analyzing training properties of corpus-specific \ac{gr} models.
Incorporating learnable decoding strategies during the training of a differentiable search index may also be of interest in this field.

%% file: Sections/10-acknowledgements.tex
\begin{acks}
    Shiguang Wu gratefully acknowledges Yixiao Yu for insightful discussions regarding Lemma~\ref{lem:kl_uniform}, which were pivotal in establishing the first main result of this work.
    This research was (partially) funded by
    the Natural Science Foundation of China (62472261, 62372275, 62272274, 62202271, T2293773),
    the National Key R\&D Program of China with grant No. 2024YFC\-3307300 and No. 2022YFC\-3303004,
    the Provincial Key R\&D Program of Shandong Province with grant No. 2024CXGC010108, 
    the Natural Science Foundation of Shandong Province with grant No. ZR2024QF203, 
    the Technology Innovation Guidance Program of Shandong Province with grant No. YDZX2024088,
    the Hybrid Intelligence Center, a 10-year program funded by the Dutch Ministry of Education, Culture and Science through the Netherlands Organisation for Scientific Research, \url{https://hybrid-intelligence-centre.nl}, 
    project LESSEN with project number NWA.1389.20.183 of the research program NWA ORC 2020/21, which is (partly) financed by the Dutch Research Council (NWO), 
    project ROBUST with project number KICH3.LTP.\-20.006, which is (partly) financed by the Dutch Research Council (NWO), DPG Media, RTL, and the Dutch Ministry of Economic Affairs and Climate Policy (EZK) under the program LTP KIC 2020-2023,
    and
    the FINDHR (Fairness and Intersectional Non-Discrimi\-nation in Human Recommendation) project that received funding from the European Union’s Horizon Europe research and innovation program under grant agreement No 101070212.
    All content represents the opinion of the authors, which is not necessarily shared or endorsed by their respective employers and/or sponsors.
\end{acks}

%% file: Sections/A-appendix.tex
\section{On the magnitude of vocabulary size}\label{sec:complexity}
In this work, we use a sequence of tokens as the \ac{docid} to represent the document as a general case.
In order for the \ac{gr} model to generate the \ac{docid} of relevant documents, the \ac{docid} faithfully representing the semantic information is highly perfered.
Therefore, the vocabulary size of the \ac{docid} is a key factor to determine the capacity of the \ac{gr} model.
In this section, we will give a brief perspective of its magnitude by considering the average information content in real-world documents.

\Ac{bpb} is a widely used metric to measure the information content of a document.
It is defined as the number of bits required to encode the content in a lossless way.
Let $\alpha$ be the \ac{bpb} of general English text, and $n$ be the length (in bytes) of a document.
The information content of the document is $n \alpha$ bits on average.
Since the vocabulary size of the \ac{gr} model is $k$, the maximum average information in a single token is $\log_2 k$ bits.
Each \ac{docid} has $m$ tokens, so there will be $m \log_2 k$ bits in total.
One would expect that the information content of the \ac{docid} should be larger than the one of the document, i.e., $m \log_2 k \ge n \alpha$, which implies $k \ge 2^{\frac{n \alpha}{m}}$.
From \citet{openai2024gpt4technicalreport} and several publications~\citep{borgeaud22languagemodels,li2024evaluatinglargelanguagemodels}, the \ac{bpb} is usually around $1$.
For a document of length $512$ bytes, and our \ac{docid} length $m$ is $32$, the vocabulary size $k$ should be about $2^{16\alpha}\approx 65,536$.
This is about the same size of the vocabulary in a language model.
The size of the complete corpus is astronomically high, and for a regular size downstream corpus, the sampling probability is approaching zero, and will raise large \ac{kl} divergence according to Section~\ref{sec:constraint}.
This is similar to the case where we want to use the language model as \ac{gr} model and some textual content as the \ac{docid}.

%% file: Sections/B-appendix.tex
\section{Lower bound of \ac{kl} divergence}\label{app:main_result}

\begin{lemma}[Lower bound of \ac{kl} divergence between Binomial and uniform distribution]\label{lem:kl_uniform}
    Let $\rv{S}_i\sim \ms{Binomial}(m,p)$, where $i\in[k]$, and $Z=\sum_{i=1}^k S_i$.
    We define a normalized distribution $P$ as
    \begin{equation}
        P_i\coloneqq\frac{S_i}{Z}, i\in[k],
    \end{equation}
    and a uniform distribution $Q$ on $\ms{supp}(P)\coloneqq \{i|S_i>0\}$ as
    \begin{equation}
        Q_i\coloneqq
        \begin{cases}
            \frac{1}{|\ms{supp}(P)|}, & \text{if } S_i>0 \\
            0, & \text{otherwise}
        \end{cases}
    \end{equation}    
    Then, we have a lower bound of the \ac{kl} divergence between $P$ and $Q$ for large $k$,
    \begin{equation}
        \KL{P}{Q}=\sum_i P_i \ln \frac{P_i}{Q_i}\gtrsim \frac{0.05}{mp}.
        \label{eq:kl}
    \end{equation}
\end{lemma}

\begin{proof}
    By the De Moivre–Laplace theorem~\citep{LyapunovCLT}, as $k$ grows large, for $n$ in the neighborhood of $mp$, we have
    \begin{equation}
        \rv{S}_i \sim \mc{N}(mp, mpq), \, p+q=1.
    \end{equation}
    Therefore, let $n=mp+\sqrt{mpq}$, we have
    \begin{equation}
        \Pr[S_i\ge n] \simeq 1 - \Phi(1). % \approx 0.16.
    \end{equation}    
    Next, we focus on $\rv{S}_i$ that deviates from the mean.
    Let $A_n\coloneqq \{S_i\ge n\}$ be the deviated elements, and $Y_n=|A_n|$, we have
    \begin{align}
        P(A_n)
        & {} = \frac{\sum_i S_i\ind[S_i\ge n]}{Z}
        \ge \frac{nY_n}{Z}, \text{and } \\
        Q(A_n)
        & {} =\frac{\sum_i \ind[S_i\ge n]}{\sum_i \ind[S_i\ge 1]}
        =\frac{Y_n}{Y_1}.
    \end{align}
    Since 
    \begin{align}
        \E[Z]\ & = kmp,\\
        \E[Y_n] & = k\Pr[S_i\ge n] = O(k), \\
        \E[Y_1] & = k(1-\Pr[S_i=0])\simeq k\left(1-e^{-k^s}\right),
    \end{align}
    by the multiplicative Chernoff bound, with probability at least $1-\exp(-\delta^2\Omega(k))$, % $1-6\exp(-\delta^2\Omega(k))$,
    $|Y_n-\E[Y_n]|\le \delta\ \E[Y_n]$,
    $|Y_1-\E[Y_1]|\le \delta\ \E[Y_1]$, and
    $|Z-mpq|\le \delta\ mpq$,
    we set $\delta=o(1)$ and $\delta^2k=k^{O(1)}$, and have
    \begin{align}
        d_\ms{TV}(P,Q)
        & {} \ge |P(A_n)-Q(A_n)| \\
        & {} \ge \frac{nY_n}{Z}-\frac{Y_n}{Y_1} \\
        & {} \ge \frac{(1-\delta)n\E[Y_n]}{(1+\delta)kmp}-\frac{(1+\delta)\E[Y_n]}{(1-\delta)\E[Y_1]} \\
        & {} \simeq \frac{1-\Phi(1)}{\sqrt{mp}}.
    \end{align}
    where $d_\ms{TV}(P,Q)=\sup_{A\subseteq [k]}|P(A)-Q(A)|$ is the total variation distance of two distributions.
    
    Lastly, we use Pinsker's inequality~\citep{pinskerinequality} to give the asymptotic lower bound in Eq.~\ref{eq:kl},
    \begin{align}
        \KL{P}{Q}
        & \ge 2 d_\ms{TV}(P,Q)^2\gtrsim \frac{0.05}{mp}.\qedhere
    \end{align}
\end{proof}

\begin{theorem}[Lower bound of \ac{kl} divergence for uniform relevance distribution.]\label{the:kl_uniform}
    Let $r=m-1$ and the sampling probability $p=\frac{1}{k^{r-s}}$.
    $\Pr(\cdot)$ is a uniform distribution.
    We have, for the \ac{kl} divergence in Eq.~\ref{eq:kl_error},
    \begin{equation}
        \KL{\Pr(\cdot|C)}{\Pr(\cdot|C_i)}\gtrsim \frac{0.05}{k^{s}}.
    \end{equation}
    Here $s$ is small and hence the right hand side converges slowly with respect to $k$.
    When $s=0$, we have a constant lower bound $0.05$.
\end{theorem}

\begin{proof}
    We consider each possible token $d_1\in[k]$ at first position.
    Since the selection of each document follows $\ms{Bernoulli}(p)$, the number of documents selected with first token being $d_1$ follows $\ms{Binomial}(k^r,p)$.
    As the \ac{gr} model can only consider the constraint in the current step, it will return a uniform distribution over the valid $d_1$ tokens.
    By revoking Lemma~\ref{lem:kl_uniform}, we have the lower bound.
\end{proof}

\begin{theorem}[Lower bound of \ac{kl} divergence for general relevance distribution]\label{the:kl_arbitrary}
    Let $S_{ij}$ be a weighted Bernoulli random variable, with parameter $p$ and weight $w_{ij}$, where $i\in[k]$ and $j\in[n]$.
    Suppose for some $\delta>0$, $\{S_{ij}\mid j\in[n]\}$ satisfy Lyapunov's condition~\citep{LyapunovCLT}, i.e.,
    \begin{equation}
        \lim_{n\to \infty} \frac{pq^{2+2\delta}+p^{2+2\delta}q}{\left(\sum_jw_{ij}^2pq\right)^{1+\delta}} \sum_jw_{ij}^{2+2\delta}=0.
    \end{equation}
    
    We define $P$ and $Q$ similar to Theorem~\ref{the:kl_uniform}, as
    \begin{equation}
        P[i]=\frac{S_i}{Z}, \quad Q[i]=\frac{w_i}{W},
    \end{equation}
    where $Z=\sum_{i=1}^k S_i$ and $W=\sum_{i=1}^k w_i$.
    We have a lower bound of the \ac{kl} divergence between $P$ and $Q$ for large $k$,
    \begin{equation}
        \KL{P}{Q}\gtrsim \frac{0.05\E^2[A_i]}{p},
    \end{equation}
    where $A_i^2=\sum_j w_{ij}^2/w_i^2$.
\end{theorem}

\begin{proof}
    Here we use the Lyapunov central limit theorem~\citep{LyapunovCLT}, to approximate the distribution of $W_i'$.
    We have
    \begin{equation}
        S_i\sim \mc{N}\bigl( w_i p, pqw_i^2A_i^2 \bigr) =\mc{N}(\mu_i, \sigma_i^2), \, p+q=1.
    \end{equation}

\noindent%
    As we have done in Theorem~\ref{the:kl_uniform}, we choose a subset of $[k]$ to compute a lower bound of total variation distance.
    Let $I=\{i \mid S_i \ge \mu_i+\sigma_i\}$, $\delta=o(1)$ and $k\delta^2=k^{O(1)}$, we have, with probability at least $1-3\exp\bigl(\frac{-2\delta^2k^2}{B^2}\bigr)$, where $B=\sum_{i\in[k]} w_i^2$,
    \begin{align}
        d_\ms{TV}(P,Q) & \ge |P(I)-Q(I)| \\
        & \ge \frac{\sum_{i\in I} S_i}{Z}-\frac{\sum_{i\in I} w_i}{W} \\
        & \gtrsim \frac{0.16}{\sqrt{p}}\,\frac{\sum_{i\in[k]} w_iA_i}{\sum_{i\in[k]w_i}}.
    \end{align}
    If we set $W=1$, it can be simplified as $\frac{0.16\E[A_i]}{\sqrt{p}}$.
\end{proof}

%% file: Sections/C-appendix.tex
\section{Recall performance using beam search}\label{app:recall}

\begin{theorem}[Top-$\lambda k$ recall and top-$1$ precision of data model in Section~\ref{sub:overtaken}]\label{the:recall}
Suppose we have $k$ branches, each with $n=k^{m-1}$ documents (leaves).
We randomly select $\lambda k$ docs from all branches as the relevant ones, where $\lambda \ll 1$.
We assign each non-relevant doc a score uniformly at random from $[-1, 1]$, and each relevant doc from $[\delta-\Delta, \delta+\Delta]$, where $\delta=0.5 \log (0.8 n)$, and $\Delta=O(0.5 \log \log k)$.
We then take the exponential of the score for each doc and use the sum of the scores as the ranking score for each branch.
Then, we have the following results with high probability:
\begin{enumerate}
    \item The top-$\lambda k$ recall is lower than $0.5 + \max\{0.65 - 0.15/\lambda, 0\}$.
    \item The top-$1$ precision is $1$.
\end{enumerate}

\end{theorem}

\begin{proof}
    First, we show the distribution of values for non-relevant branches.
    For each non-relevant branch, the score $S$ is the sum of $n$ i.i.d. random variables from $\rv{Y}_i=e^{\rv{X}_i}$, where $\rv{X}_i \sim \ms{Uniform}(-1, 1)$.
    As $\E[\rv{Y}] \approx 1.175$, and $\Var[\rv{Y}] \approx 0.4$, we have $\E[\rv{S}] = 1.175 n$ and $\Var[\rv{S}] = 0.4 n$.
    According to the central limit theorem, the distribution of $S$ is approximately $\mc{N}(1.175 n, 0.4 n)$.
    Therefore, for two non-relevant branches, the difference of their scores is approximately $\mc{N}(0, 0.8 n)$, and $\Pr[S_1 - S_2 \ge \sqrt{0.8 n}] \approx 0.15$.

    Next, we estimate the number of relevant branches.
    As the proportion $p=\frac{\lambda k}{kn}$ is small, and $np = \lambda$, we can approximate the number of relevant docs, $\#\rv{R}$, in each branch as a Poisson distribution with parameter $\lambda$.
    Then we have
    $\Pr[\#\rv{R}=0] = e^{-\lambda} \approx 1-\lambda$, and $\Pr[\#\rv{R}=1] = \lambda e^{-\lambda} \approx \lambda,$
    Therefore, there are aproximately $\lambda k$ relevant branches and $k-\lambda k$ non-relevant branches.

    We consider a barrier $B=1.175 n + \sqrt{0.8 n}$ and the event that some non-relevant branches are above the barrier and half of relevant branches are below the barrier.
    According to the data model, with high probability, half of the relevant documents have a score below $\sqrt{0.8 n}$.
    Since each non-relevant branch has a probability of $0.15$ to exceed another by $\sqrt{0.8 n}$, we have that, with high probability, at least $0.15 (k-\lambda k)$ non-relevant branches will exceed one of the low-score relevant branches.
    Thus at most $0.5\lambda k + \max\{0.5\lambda k - 0.15 (k-\lambda k), 0\}$ relevant branches will be in the top-$\lambda k$.
    Then the recall is at most $0.5 + \max\{0.65 - 0.15/\lambda, 0\}$.
    
    For the top-$1$ precision, the maximum score of the relevant documents can approach $\delta+\Delta$ w.h.p.
    And for a normal distribution $\rv{D}$ with mean $0$ and variance $\delta^2=0.8 n$, $\Pr\left[\frac{\rv{D}}{\delta}\ge \epsilon\right] \le O(\frac{1}{\epsilon}e^{-\epsilon^2/2})$.
    If we let $\epsilon=e^\Delta$, we have that the probability there is a non-relevant branch exceeding by $e^\Delta\sqrt{0.8 n}$ is $o(1)$.
    Then w.h.p. the top-$1$ precision is $1$.
\end{proof}

\begin{remark}
    In fact, the recall will be lower than $0.5$ if $\lambda$ is small enough because more non-relevant branches will be much higher than the barrier.
    The result mainly comes from the carefully designed score of relevant documents which is linear in $n$.
    It may not hold for extremely skewed distribution of scores, e.g., the relevant score is exponentially large, which actually corresponds to the ``amplification'' discussed in Section~\ref{sub:solution}.
    % It is also possible to allow for larger range of non-relevant scores to increase the recall, and we would still be able to get a similar upper bound with careful analysis.
\end{remark}